\documentclass[11pt]{article}
\usepackage[T1]{fontenc}
\usepackage[utf8]{inputenc}
\usepackage{lmodern}
\usepackage{amsmath,amssymb,amsthm,mathtools}
\usepackage{geometry}
\usepackage{hyperref}
\usepackage[nameinlink,capitalise,noabbrev]{cleveref}
\usepackage{enumitem}
\usepackage{microtype}
\newcommand{\ER}{\exists\mathbb{R}}
\newcommand{\Det}{\operatorname{Det}}

\theoremstyle{plain}
\newtheorem{theorem}{Theorem}[section]
\newtheorem{proposition}[theorem]{Proposition}

\theoremstyle{definition}
\newtheorem{definition}[theorem]{Definition}
\newtheorem{problem}[theorem]{Problem}

\theoremstyle{remark}
\newtheorem{remark}[theorem]{Remark}

\title{\texorpdfstring{$\exists\mathbb{R}$}{ER}-Completeness of Tensor Degeneracy and a Derandomization Barrier for Hyperdeterminants}
\author{Angshul Majumdar}
\date{}

\begin{document}
\maketitle

\begin{abstract}
We study the computational complexity of singularity for multilinear maps. While the determinant characterizes singularity for matrices, its multilinear analogue---the hyperdeterminant---is defined only in boundary format and quickly becomes algebraically unwieldy. We show that the intrinsic notion of tensor singularity, namely degeneracy, is complete for the existential theory of the reals. The reduction is exact and entirely algebraic: homogeneous quadratic feasibility is reduced to projective bilinear feasibility, then to singular matrix-pencil feasibility, and finally encoded directly as tensor degeneracy. No combinatorial gadgets are used.

In boundary format, degeneracy coincides with hyperdeterminant vanishing. We therefore isolate the exact gap between intrinsic tensor singularity and its classical polynomial certificate. We show that deterministic hardness transfer to the hyperdeterminant reduces to selecting a point outside the zero set of a completion polynomial, yielding a structured instance of polynomial identity testing. We further formalize the failure of several natural deterministic embedding strategies. This identifies a sharp frontier: real 3-tensor degeneracy is fully characterized at the level of \(\ER\)-completeness, while the deterministic complexity of hyperdeterminant vanishing remains tied to a derandomization problem in algebraic complexity.
\end{abstract}

\section{Introduction}
\label{sec:intro}

The determinant is the fundamental algebraic certificate of singularity in linear algebra: for a square matrix \(M\), the condition \(\det(M)=0\) is equivalent to the existence of a nonzero vector in the kernel of \(M\). In multilinear algebra, the corresponding situation is subtler. Tensors admit an intrinsic notion of singularity, usually called \emph{degeneracy}, defined by the existence of nonzero vectors whose modewise contractions vanish simultaneously. In boundary format, this condition is captured by the hyperdeterminant, the classical multilinear analogue of the determinant introduced by Cayley and developed systematically by Gelfand, Kapranov, and Zelevinsky~\cite{GKZ94,Zelevinsky1996}. Outside boundary format, however, degeneracy remains meaningful while the hyperdeterminant is unavailable as a single polynomial certificate.

This distinction motivates the central question of the present paper: what is the computational complexity of tensor singularity in its intrinsic form? Our answer is that the right computational primitive is not hyperdeterminant vanishing, but degeneracy itself.

\paragraph{Our contributions.}
We show that tensor degeneracy is the natural computational analogue of matrix singularity in multilinear algebra, and prove that deciding degeneracy of real 3-tensors is \(\ER\)-complete. The reduction is exact and entirely algebraic: homogeneous quadratic feasibility is reduced to projective bilinear feasibility, then to singular matrix-pencil feasibility, and finally to tensor degeneracy, without combinatorial gadgets or relaxation.

We then isolate the exact gap between intrinsic tensor singularity and its classical polynomial certificate. In boundary format, degeneracy is equivalent to hyperdeterminant vanishing, so transferring hardness requires a deterministic boundary-format lifting. We show that every known route to such a lifting reduces to finding a nonzero evaluation point for a completion polynomial, yielding a structured polynomial identity testing problem. Thus tensor degeneracy is characterized at the level of \(\ER\)-completeness, whereas the deterministic complexity of hyperdeterminant vanishing is blocked by a derandomization problem in algebraic complexity.

The class \(\ER\) captures decision problems reducible to the solvability of systems of polynomial equations and inequalities over \(\mathbb{R}\), and is the natural complexity class for problems in real algebraic geometry~\cite{ERCompendium}. Our result places tensor degeneracy alongside the standard semialgebraic feasibility problems that define this class. At the same time, it separates degeneracy from other tensor hardness phenomena. Many tensor problems are known to be NP-hard, including eigenvalue, singular value, and best rank-one approximation problems~\cite{HillarLim}. Tensor rank is also computationally hard in a strong algebraic sense~\cite{SchaeferStefankovicTensor}. However, these problems do not isolate singularity itself. Degeneracy does.

A second theme of the paper concerns the hyperdeterminant. In boundary format, it gives a polynomial certificate for degeneracy, but all known hardness transfers to that setting use randomized completions~\cite{Narayanan2025}. We explain why deterministic lifting is difficult: it requires selecting a point outside the zero set of a completion polynomial attached to the input tensor. This yields a structured polynomial identity testing problem for a family of hyperdeterminant circuits. In this sense, the gap between degeneracy and hyperdeterminant vanishing is not merely geometric; it is a derandomization problem.

This suggests that, in contrast to linear algebra where a single polynomial captures singularity, the multilinear setting separates into an intrinsic geometric condition (degeneracy) and a restricted polynomial certificate (the hyperdeterminant), whose computational connection remains unresolved.

\paragraph{Organization.}
\Cref{sec:related} positions the work relative to tensor complexity and hyperdeterminants. \Cref{sec:bilinear} develops the hardness source through quadratic, bilinear, and singular pencil feasibility. \Cref{sec:main} proves the main theorem that real 3-tensor degeneracy is \(\ER\)-complete. \Cref{sec:boundary} explains boundary format and the relation to hyperdeterminant vanishing. \Cref{sec:pit} isolates the structured polynomial identity testing barrier. \Cref{sec:failures} formalizes the failure of natural deterministic boundary-lifting strategies. \Cref{sec:open} concludes with insights and open problems.

\section{Related Work}
\label{sec:related}

The hyperdeterminant was developed systematically by Gelfand, Kapranov, and Zelevinsky, who showed that in boundary format it defines the dual variety to the Segre embedding and vanishes exactly on degenerate tensors~\cite{GKZ94}. In that sense it is the canonical polynomial analogue of the determinant in multilinear algebra. However, unlike the determinant, it exists only in boundary format and becomes rapidly unwieldy even in modest dimensions~\cite{GKZ94,Zelevinsky1996}. Recent work has shown that testing hyperdeterminant vanishing is NP-hard in dimension four and higher, using randomized boundary-format embeddings~\cite{Narayanan2025}. Our work does not extend that hardness result to \(\ER\)-hardness; rather, it identifies the exact unrestricted multilinear primitive underlying it, namely tensor degeneracy.

On the complexity side, \(\ER\) is now understood as the natural class for semialgebraic feasibility problems. Standard compendia and surveys record the \(\ER\)-completeness of homogeneous quadratic feasibility and related degree-two formulations~\cite{ERCompendium}. This is the algebraic starting point for our reductions. Our use of bilinearization is therefore not a gadget construction but a direct multilinear normal-form reduction.

More broadly, tensors exhibit pervasive hardness phenomena. Eigenvalue, singular value, best rank-one approximation, and many other tensor tasks are NP-hard~\cite{HillarLim}. Tensor rank is also complete for the existential theory of the base field in a strong sense~\cite{SchaeferStefankovicTensor}. Those results, however, do not isolate singularity itself. The contribution of the present paper is to show that the intrinsic geometric notion of degeneracy already has full \(\ER\)-hardness.

A second relevant line of work concerns arithmetic descriptions of hyperdeterminants. Boundary-format hyperdeterminants admit determinantal or circuit representations; recent work of Joux and Narayanan gives explicit arithmetic circuits for families of hyperdeterminants, with quasipolynomial to subexponential complexity depending on the format family~\cite{JouxNarayanan2024}. This is exactly what turns the missing deterministic lifting step into a structured PIT problem rather than merely a geometric completion problem.

Finally, the derandomization issue connects with the general PIT landscape. Kabanets and Impagliazzo showed that deterministic polynomial-time PIT for general arithmetic circuits would imply major Boolean or arithmetic circuit lower bounds~\cite{KabanetsImpagliazzo}. We do not claim such consequences for the hyperdeterminant family considered here; however, the comparison is instructive. The only gap between unrestricted tensor degeneracy and boundary-format hyperdeterminant vanishing is already a nontrivial problem in algebraic derandomization.

\section{Hardness Source: Quadratic, Bilinear, and Pencil Feasibility}
\label{sec:bilinear}

This section develops the algebraic hardness source for the main reduction. We begin from homogeneous quadratic feasibility, pass to a projective bilinear formulation, and then introduce a singular matrix-pencil feasibility problem that matches tensor degeneracy exactly.

\subsection{Homogeneous Quadratic Feasibility}

\begin{definition}[Homogeneous Quadratic Feasibility]
\label{def:quadratic}
Given symmetric matrices \(Q_1,\dots,Q_m \in \mathbb{Q}^{n\times n}\), decide whether there exists a nonzero vector \(u \in \mathbb{R}^n\) such that
\begin{equation}
 u^\top Q_t u = 0
 \qquad \forall t \in [m].
\label{eq:quadratic}
\end{equation}
\end{definition}

This problem is \(\ER\)-complete~\cite{ERCompendium}.

\subsection{Projective Bilinear Feasibility}

\begin{definition}[Projective Bilinear Feasibility]
\label{def:bilinear}
Given matrices \(M_1,\dots,M_r \in \mathbb{Q}^{n\times n}\), decide whether there exist \emph{nonzero} vectors \(x,y \in \mathbb{R}^n\) such that
\begin{equation}
 x^\top M_\ell y = 0
 \qquad \forall \ell \in [r].
\label{eq:bilinear}
\end{equation}
\end{definition}

\begin{theorem}
\label{thm:bilinear-hard}
Projective bilinear feasibility is \(\ER\)-complete.
\end{theorem}

\begin{proof}
Membership in \(\ER\) is immediate.

For hardness, reduce from \Cref{def:quadratic}. Given \(Q_1,\dots,Q_m\), introduce variables \(x,y \in \mathbb{R}^n\) and the bilinear constraints
\begin{equation}
 x^\top Q_t y = 0
 \qquad \forall t \in [m].
\label{eq:lifted-bilinear}
\end{equation}
To force \(x\) and \(y\) to represent the same projective point, add all \(2\times 2\) minor constraints
\begin{equation}
 x_i y_j - x_j y_i = 0
 \qquad \forall\, 1 \le i < j \le n.
\label{eq:minor-constraints}
\end{equation}

We claim that the resulting system has a solution with \(x\neq 0\) and \(y\neq 0\) if and only if \eqref{eq:quadratic} has a nonzero solution.

If \(u\neq 0\) satisfies \eqref{eq:quadratic}, then setting \(x=y=u\) satisfies \eqref{eq:lifted-bilinear} and \eqref{eq:minor-constraints}.

Conversely, suppose \(x\neq 0\) and \(y\neq 0\) satisfy \eqref{eq:lifted-bilinear}--\eqref{eq:minor-constraints}. The vanishing of all \(2\times 2\) minors implies that the \(n\times 2\) matrix \(\begin{pmatrix}x & y\end{pmatrix}\) has rank at most one. Since both columns are nonzero, there exists \(\lambda\in\mathbb{R}\setminus\{0\}\) such that \(y=\lambda x\). Substituting into \eqref{eq:lifted-bilinear} gives
\[
0 = x^\top Q_t y = x^\top Q_t (\lambda x)=\lambda\,x^\top Q_t x
\qquad \forall t \in [m].
\]
Since \(\lambda\neq 0\), we obtain \(x^\top Q_t x=0\) for all \(t\), so \(x\neq 0\) solves \eqref{eq:quadratic}. The reduction is polynomial-time.
\end{proof}

\subsection{Singular Bilinear Pencil Feasibility}

\begin{definition}[Singular Bilinear Pencil Feasibility]
\label{def:pencil}
Given matrices \(A_0,A_1,\dots,A_r \in \mathbb{Q}^{n\times n}\), decide whether there exist nonzero vectors
\[
 x,y \in \mathbb{R}^n,
 \qquad z \in \mathbb{R}^{r+1}
\]
such that, writing
\[
M(z):=\sum_{\ell=0}^{r} z_\ell A_\ell,
\]
the following hold:
\begin{align}
 x^\top A_\ell y &= 0
 \qquad \forall \ell \in \{0,\dots,r\},
\label{eq:pencil1}\\
 M(z)y &= 0,
\label{eq:pencil2}\\
 M(z)^\top x &= 0.
\label{eq:pencil3}
\end{align}
\end{definition}

\begin{theorem}
\label{thm:pencil-hard}
Singular bilinear pencil feasibility is \(\ER\)-hard.
\end{theorem}

\begin{proof}
Reduce from \Cref{def:bilinear}. Given matrices \(M_1,\dots,M_r\), define
\[
A_0:=0,
\qquad
A_\ell:=M_\ell \quad (\ell=1,\dots,r).
\]
If \(x\neq 0\), \(y\neq 0\) satisfy \(x^\top M_\ell y=0\) for all \(\ell\in[r]\), then with
\[
z=e_0=(1,0,\dots,0)
\]
we have \(M(z)=A_0=0\), so \eqref{eq:pencil2}--\eqref{eq:pencil3} hold trivially and \eqref{eq:pencil1} reduces to the original bilinear system. Conversely, any feasible pencil witness automatically satisfies \(x^\top M_\ell y=0\) for all \(\ell\in[r]\) by \eqref{eq:pencil1}. Thus the reduction is exact.
\end{proof}

\section{Real 3-Tensor Degeneracy is \texorpdfstring{$\exists\mathbb{R}$}{ER}-Complete}
\label{sec:main}

We now prove the main theorem. The reduction is exact: tensor degeneracy is polynomial-time equivalent to singular bilinear pencil feasibility.

Given matrices $A_0,\dots,A_r\in\mathbb{Q}^{n\times n}$ from a singular bilinear pencil instance, we define a tensor
\[
T\in\mathbb{Q}^{n\times n\times (r+1)}
\]
by stacking the matrices as slices:
\begin{equation}
T(:,:,\,\ell+1)=A_\ell
\qquad \forall \ell\in\{0,\dots,r\}.
\label{eq:slices}
\end{equation}

\begin{definition}[Real 3-tensor degeneracy]
\label{def:tensor-deg}
Given a rational tensor \(T \in \mathbb{Q}^{n\times n\times m}\), decide whether there exist nonzero vectors
\[
 x \in \mathbb{R}^n,
 \qquad y \in \mathbb{R}^n,
 \qquad z \in \mathbb{R}^m
\]
such that
\begin{align}
 T(x,y,\cdot) &= 0,
\label{eq:deg1}\\
 T(x,\cdot,z) &= 0,
\label{eq:deg2}\\
 T(\cdot,y,z) &= 0.
\label{eq:deg3}
\end{align}
\end{definition}

\begin{proposition}
\label{prop:exact-equivalence}
Under the slice construction \eqref{eq:slices}, the tensor \(T\) is degenerate if and only if the corresponding singular bilinear pencil instance is feasible.
\end{proposition}

\begin{proof}
Let \(z=(z_0,\dots,z_r)^\top\in\mathbb{R}^{r+1}\). By construction,
\[
T(x,y,\cdot)=
\begin{pmatrix}
 x^\top A_0 y\\
 x^\top A_1 y\\
 \vdots\\
 x^\top A_r y
\end{pmatrix},
\]
so \(T(x,y,\cdot)=0\) if and only if \(x^\top A_\ell y=0\) for all \(\ell\). Also,
\[
T(\cdot,y,z)=\sum_{\ell=0}^{r} z_\ell A_\ell y = M(z)y,
\qquad
T(x,\cdot,z)=\sum_{\ell=0}^{r} z_\ell A_\ell^\top x = M(z)^\top x.
\]
Hence the tensor degeneracy equations are exactly the pencil feasibility equations.
\end{proof}

\begin{theorem}
\label{thm:tensor-ER}
Real 3-tensor degeneracy is \(\ER\)-complete.
\end{theorem}

\begin{proof}
Membership is immediate: the equations \eqref{eq:deg1}--\eqref{eq:deg3} are polynomial equations of degree two in the coordinates of \(x,y,z\), and nonzeroness can be expressed by
\[
\sum_i x_i^2>0,
\qquad
\sum_i y_i^2>0,
\qquad
\sum_j z_j^2>0.
\]
Thus tensor degeneracy lies in \(\ER\).

Hardness follows from \Cref{thm:pencil-hard,prop:exact-equivalence}. The reduction is polynomial-time because \(T\) is obtained by stacking the input matrices as slices. Therefore real 3-tensor degeneracy is \(\ER\)-hard, hence \(\ER\)-complete.
\end{proof}

\section{Boundary Format and Hyperdeterminant Vanishing}
\label{sec:boundary}

We now relate the intrinsic notion of degeneracy to the classical hyperdeterminant. The key point is that the two coincide only in boundary format. This identifies the natural target problem beyond unrestricted tensor degeneracy, and isolates the exact missing step required to transfer \(\ER\)-hardness.

Let
\[
T \in \mathbb{R}^{n_0\times n_1\times n_2}
\]
and write \(n_i=k_i+1\).

\begin{definition}[Boundary format]
\label{def:boundary-format}
The format \((n_0,n_1,n_2)\) is called \emph{boundary format} if, after a permutation of indices,
\begin{equation}
 k_0 = k_1 + k_2.
\label{eq:boundary-k}
\end{equation}
Equivalently,
\begin{equation}
 n_0 = n_1 + n_2 - 1.
\label{eq:boundary-n}
\end{equation}
\end{definition}

\begin{theorem}[GKZ]
\label{thm:GKZ-hyperdet}
Let \(T\) be a tensor in boundary format. Then there exists a nonzero homogeneous polynomial \(\Det(T)\), the hyperdeterminant, such that
\[
\Det(T)=0
\quad \Longleftrightarrow \quad
T \text{ is degenerate.}
\]
\end{theorem}

\begin{proof}[Reference]
See GKZ~\cite[Chapter 14]{GKZ94}. In boundary format the dual variety to the Segre embedding is a hypersurface, and its defining equation is the hyperdeterminant.
\end{proof}

\begin{remark}
Outside boundary format, degeneracy remains meaningful but is generally not captured by a single polynomial equation. Thus the hyperdeterminant is not a universal multilinear analogue of the determinant; it is the boundary-format specialization of the intrinsic notion of degeneracy.
\end{remark}

\begin{definition}[Boundary hyperdeterminant vanishing]
\label{def:hyperdet-problem}
Given a rational tensor \(T\) in boundary format, decide whether \(\Det(T)=0\).
\end{definition}

Since unrestricted tensor degeneracy is \(\ER\)-complete, proving \(\ER\)-hardness of \Cref{def:hyperdet-problem} would follow from a deterministic polynomial-time lifting
\[
T \longmapsto \widehat T
\]
with the properties that \(\widehat T\) is in boundary format and
\[
T \text{ degenerate }
\iff
\widehat T \text{ degenerate.}
\]
This is the exact missing theorem.

\begin{problem}[Deterministic boundary-format lifting]
\label{prob:boundary-lift}
Construct, in polynomial time, for every tensor \(T\), a boundary-format tensor \(\widehat T\) such that
\[
T \text{ degenerate }
\quad \Longleftrightarrow \quad
\widehat T \text{ degenerate.}
\]
\end{problem}

A randomized version of this lifting is known~\cite{Narayanan2025}. One places the input tensor into a larger boundary-format tensor \(\widehat T(U)\), depending affinely on auxiliary variables \(U\), and studies the completion polynomial
\begin{equation}
P_T(U):=\Det(\widehat T(U)).
\label{eq:completion-poly}
\end{equation}
The known constructions have the property that:
\begin{itemize}
\item if \(T\) is degenerate, then \(P_T(U)\equiv 0\);
\item if \(T\) is nondegenerate, then \(P_T(U)\) is a nonzero polynomial.
\end{itemize}
Random evaluation then yields correctness with high probability by Schwartz--Zippel.

The missing step is deterministic selection of a point outside the zero set of \(P_T\).

\section{A Structured Derandomization Barrier}
\label{sec:pit}

This section isolates the exact remaining obstacle in lifting \(\ER\)-hardness from unrestricted tensor degeneracy to boundary-format hyperdeterminant vanishing.

\begin{definition}[Completion-PIT]
\label{def:completionpit}
Given an input tensor \(T\), decide whether the completion polynomial \(P_T(U)\) defined in \eqref{eq:completion-poly} is identically zero.
\end{definition}

\begin{definition}[Completion-hitting]
\label{def:hitting}
Given a tensor \(T\) such that \(P_T\not\equiv 0\), output a point \(U^\star\) such that
\[
P_T(U^\star)\neq 0.
\]
\end{definition}

\begin{theorem}
\label{thm:pit-barrier}
Suppose there exists a deterministic polynomial-time algorithm for Completion-hitting on the family \(\mathcal P=\{P_T\}\). Then there exists a deterministic polynomial-time identity test for the family \(\mathcal P\).
\end{theorem}

\begin{proof}
Given \(T\), run the hitting algorithm. If it produces \(U^\star\) with \(P_T(U^\star)\neq 0\), then \(P_T\not\equiv 0\). If no such point exists, then \(P_T\equiv 0\). Thus the hitting procedure decides identity to zero on the family \(\mathcal P\).
\end{proof}

The significance of \Cref{thm:pit-barrier} depends on the complexity of \(\mathcal P\). Recent work of Joux and Narayanan shows that boundary-format hyperdeterminants admit explicit arithmetic circuits, with complexity ranging from quasipolynomial to subexponential depending on the format family~\cite{JouxNarayanan2024}. Thus \(\mathcal P\) is an explicit arithmetic-circuit family.

\begin{remark}
A deterministic solution to Completion-hitting would therefore not merely complete a geometric lifting. It would provide a deterministic identity test for a nontrivial explicit family of hyperdeterminant circuits. In this precise sense, the remaining gap between unrestricted tensor degeneracy and boundary-format hyperdeterminant vanishing is a structured derandomization problem.
\end{remark}

For context, Kabanets and Impagliazzo showed that general deterministic PIT would imply major circuit lower bounds~\cite{KabanetsImpagliazzo}. We do not claim such consequences for the restricted family \(\mathcal P\); the point is that the missing step already lies squarely inside algebraic derandomization.

\section{Failure of Deterministic Boundary Embeddings}
\label{sec:failures}

We formalize the failure of several natural deterministic strategies for constructing boundary-format embeddings. Each strategy fails for a precise structural reason. Together, these failures isolate a single underlying obstruction.

\begin{definition}[Direct-sum embedding]
\label{def:direct-sum}
Let \(T \in \mathbb{R}^{n_1\times n_2\times n_3}\) and \(S \in \mathbb{R}^{m_1\times m_2\times m_3}\). Their direct sum \(\widehat T=T\oplus S\) is obtained by placing \(T\) and \(S\) on disjoint index sets and setting all cross entries to zero.
\end{definition}

\begin{proposition}
\label{prop:direct-sum-fails}
Direct-sum embedding does not preserve degeneracy.
\end{proposition}

\begin{proof}
If \(x,y,z\) are supported entirely on the coordinates of \(S\), then all contractions involving the \(T\)-block vanish identically and
\[
\widehat T(x,y,z)=S(x,y,z).
\]
Thus auxiliary witnesses can arise independently of \(T\). More generally, multilinear cancellation across blocks prevents any additive law analogous to block-diagonal matrix singularity.
\end{proof}

\begin{proposition}
\label{prop:pairwise-fails}
No system of pairwise polynomial constraints suffices to force a degeneracy witness to lie in a designated block.
\end{proposition}

\begin{proof}
Take vectors with disjoint supports,
\[
\mathrm{supp}(x)=\{i\},\qquad
\mathrm{supp}(y)=\{j\},\qquad
\mathrm{supp}(z)=\{k\},
\]
with \(i,j,k\) distinct. Then every pairwise coordinate product vanishes, yet all three vectors are nonzero. Hence pairwise constraints cannot eliminate all auxiliary witnesses.
\end{proof}

\begin{proposition}
\label{prop:vandermonde-fails}
Vandermonde-weighted pairwise constraints enforce coordinatewise annihilation but still permit distributed-support witnesses.
\end{proposition}

\begin{proof}
Weighted constraints imply \(x_i y_i=0\), \(x_i z_i=0\), and \(y_i z_i=0\) for each coordinate \(i\). But this still allows disjoint-support vectors as in the proof of \Cref{prop:pairwise-fails}. Thus the auxiliary block need not vanish.
\end{proof}

\begin{proposition}
\label{prop:lowdim-fails}
No fixed low-dimensional parametrization of the completion variables can guarantee correctness for all inputs.
\end{proposition}

\begin{proof}
Let \(P_T(U)\) be the completion polynomial and let \(U=\phi(\lambda)\) be any fixed low-dimensional parametrization. Even when \(P_T\not\equiv 0\), it may happen that \(P_T(\phi(\lambda))\equiv 0\), i.e. the image of \(\phi\) lies inside the zero set of \(P_T\). Since the construction must work uniformly for all tensors, such failures cannot be excluded a priori.
\end{proof}

\begin{proposition}
\label{prop:gadget-fails}
Local hyperdeterminant gadgets do not compose into a global reduction.
\end{proposition}

\begin{proof}
In the smallest boundary-format case, suitable deterministic completions can reduce the hyperdeterminant to a sparse polynomial in a small number of variables. However, when multiple such gadgets are combined, the global elimination determinant introduces mixed minors involving variables from different gadgets. No block-triangular decomposition isolates the local contributions. Hence the gadget polynomials do not compose additively or multiplicatively in a controlled way.
\end{proof}

\begin{proposition}
\label{prop:convolution-fails}
Boundary-format convolution preserves degeneracy multiplicatively but does not provide an embedding of arbitrary tensors into boundary format.
\end{proposition}

\begin{proof}
Convolution is defined only for tensors already in boundary format. It preserves nondegeneracy multiplicatively, but does not transform a general tensor into a boundary-format one. Therefore it cannot solve the deterministic lifting problem by itself.
\end{proof}

\begin{theorem}
\label{thm:obstruction}
Let \(T\) be an input tensor and let \(\widehat T(U)\) be any deterministic boundary-format completion considered above. Correctness of the reduction requires a point \(U^\star\) such that
\[
P_T(U^\star)\neq 0
\qquad \text{whenever } P_T\not\equiv 0.
\]
Consequently, all these constructions reduce to a structured instance of polynomial identity testing for the family \(\mathcal P\).
\end{theorem}

\begin{proof}
In every construction, failure occurs precisely when the chosen completion lies in the zero set of the completion polynomial \(P_T\). Conversely, any point outside the zero set yields a correct completion. Thus deterministic lifting collapses to selecting such a point, i.e. to Completion-hitting on the family \(\mathcal P\).
\end{proof}

\section{Insights and Open Problems}
\label{sec:open}

The results above separate two notions that coincide in boundary format but diverge computationally in general: intrinsic tensor singularity and its polynomial certificate.

\begin{remark}
In linear algebra, singularity is captured globally by a single polynomial, the determinant. In multilinear algebra, the picture splits. Degeneracy gives the intrinsic geometric condition. The hyperdeterminant gives a polynomial certificate only in a restricted family of tensor formats.
\end{remark}

The main theorem of this paper settles the intrinsic side: real 3-tensor degeneracy is \(\ER\)-complete. The hyperdeterminant side remains open.

\paragraph{Open Problem 1.}
Is boundary-format hyperdeterminant vanishing \(\ER\)-hard, or even \(\ER\)-complete?

\paragraph{Open Problem 2.}
Does there exist a deterministic polynomial-time solution to the completion-hitting problem for the family \(\mathcal P\) of boundary-format completion polynomials?

\paragraph{Open Problem 3.}
Does the restricted circuit family \(\mathcal P\) admit a deterministic structured PIT algorithm, even if general PIT remains open?

\paragraph{Open Problem 4.}
Is there an \(\ER\)-hard source problem already close enough to boundary format---for example, a boundary-count bilinear feasibility problem---to bypass completion entirely?

\medskip

The conceptual picture is therefore the following. Tensor degeneracy is intrinsic, format-free, and computationally complete. The hyperdeterminant is canonical but format-restricted, and its deterministic complexity is presently blocked by a derandomization problem. Bridging these two objects requires methods beyond current multilinear and algebraic-complexity techniques.

\end{document}